\documentclass{article}
\usepackage{makeidx}
\usepackage{amsthm,framed,amssymb}
\usepackage[sumlimits]{amsmath}
\usepackage{amsfonts,enumerate}
\usepackage{color,url}
\usepackage[margin=1.9cm]{geometry}
\usepackage{lineno}
\usepackage{graphicx}
\usepackage{framed}
\usepackage{todonotes}

\newtheorem{theorem}{Theorem}
\newtheorem{lemma}[theorem]{Lemma}
\newtheorem{proposition}[theorem]{Proposition}
\newtheorem{remark}[theorem]{Remark}

\newcommand{\rem}[1]{}

\numberwithin{equation}{section}
\numberwithin{figure}{section}
\begin{document}

\title{Wave breaking for the Stochastic Camassa-Holm equation\thanks{%
Paper submitted to the Physica D special issue \textit{Nonlinear Partial
Differential Equations in Mathematical Fluid Dynamics} dedicated to Prof.
Edriss S. Titi on the occasion of his 60th birthday. Work partially
supported by the EPSRC Standard Grant EP/N023781/1.}}
\author{Dan Crisan\thanks{%
Department of Mathematics, Imperial College, London SW7 2AZ, UK. Email:
d.crisan@ic.ac.uk} \ and Darryl D Holm\thanks{%
Department of Mathematics, Imperial College, London SW7 2AZ, UK. Email:
d.holm@ic.ac.uk}}
\date{ }
\maketitle

\begin{abstract}
We show that wave breaking occurs with positive probability for the
Stochastic Camassa-Holm (SCH) equation. This means that temporal
stochasticity in the diffeomorphic flow map for SCH does not prevent the
wave breaking process which leads to the formation of peakon solutions.
{We conjecture} that the time-asymptotic solutions of SCH {will} consist of emergent wave trains of peakons moving along stochastic space-time paths.
\end{abstract}



\makeatother



\section{The deterministic Camassa-Holm (CH) equation}

The deterministic CH equation, derived in \cite{CH1993}, is a nonlinear
shallow water wave equation for a fluid velocity solution whose profile $%
u(x,t)$ and its gradient both decay to zero at spatial infinity, ${%
|x|\to\infty}$, on the real line $\mathbb{R}$. Namely, 
\begin{equation}
u_{t}-u_{xxt}+3uu_{x}=2u_{x}u_{xx}+uu_{xxx}\,,  \label{CH-eqn1}
\end{equation}%
where subscripts $t$ (resp. $x$) denote partial derivatives in time (resp.
space). This nonlinear, nonlocal, completely integrable PDE may be written
in \emph{Hamiltonian form} for a momentum density $m:=u-u_{xx}$ undergoing
coadjoint motion, as \cite{CH1993} 
\begin{equation}
m_{t} = \{m,h(m)\} =-\,(\partial _{x}m+m\partial _{x})\frac{\delta h}{\delta
m}\,,  \label{CH-eqn1*}
\end{equation}%
which is generated by the Lie-Poisson bracket 
\begin{equation}
\{f,h\} (m) =-\,\int \frac{\delta f}{\delta m} \big(\partial _{x}m+m\partial
_{x})\frac{\delta h}{\delta m}\,dx  \label{CH-LPB}
\end{equation}
and Hamiltonian function 
\begin{equation}
h(m)=\frac{1}{2}\int_{\mathbb{R}}mK\ast m\,dx=\frac{1}{2}\int_{\mathbb{R}%
}u^{2}+u_{x}^{2}\,dx=\frac{1}{2}\Vert u\Vert _{H^{1}}^{2}=const.
\label{H1norm-1D}
\end{equation}%
Here, $K\ast m:=\int K(x,y)\,m(y,t)dy$ denotes convolution of the momentum
density $m$ with the Green's function of the Helmholtz operator $%
L=1-\partial _{x}^{2}$, so that 
\begin{equation}
\frac{\delta h}{\delta m}=K\ast m=u\quad \text{with}\quad K(x-y)=\frac{1}{2}%
\exp (-|x-y|)\,.  \label{Greens-relation}
\end{equation}%
Alternatively, the CH equation \eqref{CH-eqn1} may be written in advective
form as 
\begin{equation}
u_{t}+uu_{x} =-\,\partial _{x}\Big(K\ast \big(u^{2}+\frac{1}{2}u_{x}^{2}\big)%
\Big) =-\,\partial _{x}\int_{\mathbb{R}}\frac{1}{2}\exp (-|x-y|)\left(
u^{2}(y,t)+\frac{1}{2}u_{y}^{2}(y,t)\right) \,dy\,.  \label{CH-eqn3}
\end{equation}

The deterministic CH equation admits signature solutions representing a wave
train of peaked solitons, called \emph{peakons}, given by 
\begin{equation}  \label{peakontrain-soln}
u(x,t)=\frac12\sum_{a=1}^Mp_a(t)\mathrm{e}^{-|x-q_a(t)|} = K*m \,,
\end{equation}
which emerge from smooth confined initial conditions for the velocity
profile. Such a sum is an \emph{exact solution} of the CH equation (\ref%
{CH-eqn1}) provided the time-dependent parameters $\{p_a\}$ and $\{q_a\}$, $%
a=1,\dots,M$, satisfy certain canonical Hamiltonian equations, to be
discussed later. In fact, the peakon velocity wave train in %
\eqref{peakontrain-soln} is the \emph{asymptotic solution} of the CH
equation for any spatially confined $C^1$ initial condition, $u(x,0)$.

\begin{figure}[h!]
\begin{center}
\includegraphics[width=.9\textwidth,angle=0]{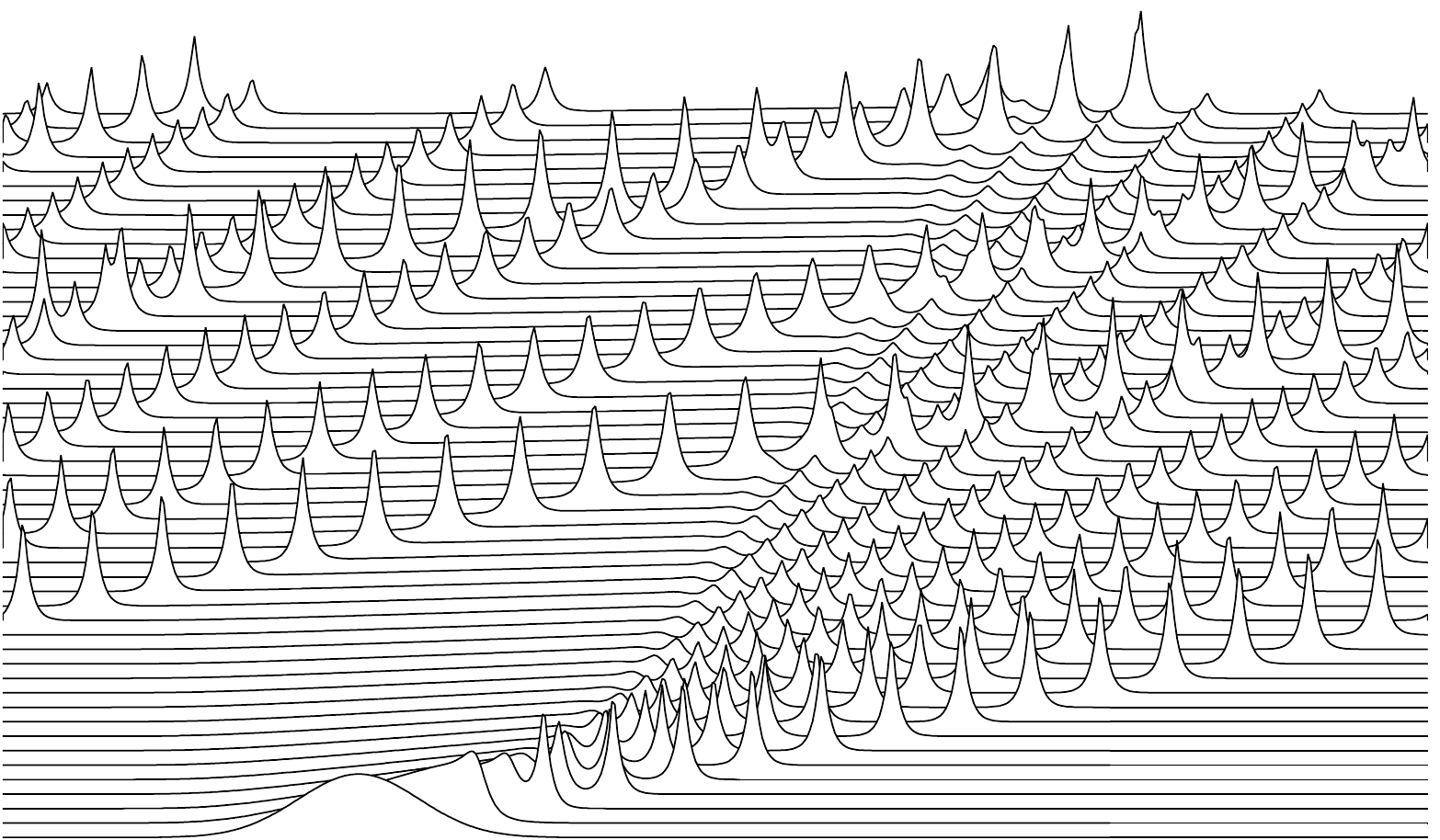}
\end{center}
\caption{{\protect\footnotesize Under the evolution of the CH equation (%
\protect\ref{CH-eqn1}), an ordered \emph{wave train of peakons} emerges from
a smooth localized initial condition (a Gaussian). The speeds are
proportional to the heights of the peaks. The spatial profiles of the
velocity at successive times are offset in the vertical to show the
evolution. The peakon wave train eventually wraps around the periodic
domain, thereby allowing the faster peakons which emerge earlier to overtake
slower peakons emerging later from behind in collisions that conserve
momentum and preserve the peakon shape but cause phase shifts in the
positions of the peaks, as discussed in \protect\cite{CH1993}.}}
\label{peakon_figure}
\end{figure}


\begin{remark}
The peakon-train solutions of CH represent an \emph{emergent phenomenon}. A
wave train of peakons emerges in solving the initial-value problem for the
CH equation (\ref{CH-eqn1}) for essentially any spatially confined initial
condition. An example of the emergence of a wave train of peakons from a
Gaussian initial condition is shown in Figure \ref{peakon_figure}. 
\index{emergent phenomenon! peakon wave train}
\end{remark}

\begin{remark}
By equation (\ref{Greens-relation}), the momentum density corresponding to
the peakon wave train \eqref{peakontrain-soln} in velocity is given by a sum
over delta functions in momentum density, representing the \emph{singular
solution}, 
\begin{equation}  \label{m-delta-N}
m(x,t) = \sum_{a=1}^M p_a(t)\,\delta(x-q_a(t)) \,,
\end{equation}
in which the \emph{delta function} $\delta(x-q)$ is defined by 
\begin{equation}
f(q) = \int f(x)\delta(x-q)\,\mathrm{d}x \,,  \label{delta-def}
\end{equation}
for an arbitrary smooth function $f$. Physically, the relationship %
\eqref{m-delta-N} represents the dynamical coalescence of the CH momentum
density into particle-like coherent structures (Young measures) which
undergo elastic collisions as a result of their nonlinear interactions.
Mathematically, the singular solutions of CH are captured by recognizing
that the singular solution ansatz \eqref{m-delta-N} itself is an equivariant
momentum map from the canonical phase space of $M$ points embedded on the
real line, to the dual of the vector fields on the real line. Namely, 
\begin{equation}  \label{momap}
m: T ^{\ast} 
{\rm Emb}(\mathbb{Z}, \mathbb{R}) \rightarrow \mathfrak{X}(\mathbb{R}%
)^{\ast}.
\end{equation}
This momentum map property explains, for example, why the singular solutions %
\eqref{m-delta-N} form an invariant manifold for any value of $M$ and why
their dynamics form a canonical Hamiltonian system, \cite{HoMa2005}.
\end{remark}

The complete integrability of the CH equation as a Hamiltonian system
follows from its isospectral problem.

\begin{theorem}[Isospectral problem for CH \protect\cite{CH1993}]
The CH equation in \eqref{CH-eqn1} follows from the compatibility conditions
for the following CH isospectral eigenvalue problem and evolution equation
for the real eigenfunction $\psi(x,t)$, 
\begin{align}
\psi_{xx} &= \left(\frac14 - \frac{m}{2\lambda} \right)\psi \,,
\label{CHevprob} \\
\partial_t\psi &= -(\lambda + u) \psi_x + \frac12 u_x\psi \,,
\label{CHevoleqn}
\end{align}
with real isospectral parameter, $\lambda$.
\end{theorem}

\begin{proof}
By direct calculation, equating cross derivatives $\partial_t\psi_{xx} = \partial_x^2\partial_t\psi$ using equations \eqref{CHevprob} and \eqref{CHevoleqn} implies the CH equation in \eqref{CH-eqn1}, provided $d\lambda/dt=0$. 
\end{proof}

\begin{remark}
The complete integrability of the CH equation as a Hamiltonian system and
its soliton paradigm explain the emergence of peakons in the CH dynamics.
Namely, their emergence reveals the initial condition's soliton (peakon)
content.
\end{remark}



\subsection{Steepening Lemma: the mechanism for peakon formation}

\label{steepen-lemma}

In the following we will continue working on the entire real line $\mathbb{R}
$, although similar results are also available for a periodic domain with
only minimal effort. We use the notation $\|u\|_2$, $\|u\|_{1,2}$ and $%
\|u\|_\infty$ to denote, respectively, 
\begin{equation*}
\|u\|_2:=\int_{-\infty}^\infty \!\!\left(u^2\right)\mathrm{d} y \,,\quad
\|u\|_{1,2}:=\int_{-\infty}^\infty \!\!\left(u^2+\frac{1}{2}u_y^2\right)%
\mathrm{d} y \,,\quad\hbox{and}\quad \|u\|_\infty := \sup_{x\in \mathbb{R}%
}\|u(x)\| \,.
\end{equation*}

\begin{remark}[Local well-posedness of CH]
As reviewed in \cite{HoMa2005}, the deterministic CH equation \eqref{CH-eqn1}
is locally well posed on $\mathbb{R}$, for initial conditions in $H^{s}$
with $s>3/2$. In particular, with such initial data, CH solutions are $%
C^{\infty }$ in time and the Hamiltonian $h(m)$ in \eqref{H1norm-1D} is
bounded for all time, 
\begin{equation*}
h:=\Vert u(\cdot ,t)\Vert _{1,2}<\infty \,.
\end{equation*}%
In fact, CH solutions preserve the Hamiltonian in \eqref{H1norm-1D} given by
the $\Vert u(\cdot ,t)\Vert _{1,2}$ norm 
\begin{equation}
\Vert u(\cdot ,t)\Vert _{1,2}=h=constant,\ \ \ \mathrm{for\ \ all}\ \ x\in 
\mathbb{R}.  \label{constantnorm}
\end{equation}%
By a standard Sobolev embedding theorem, \eqref{constantnorm} also implies
the useful relation that 
\begin{equation}
M:=\sup_{t\in \lbrack 0,\infty )}\Vert u(\cdot ,t)\Vert _{\infty }<\infty .
\label{M}
\end{equation}
\end{remark}

The mechanism for the emergent formation of the peakons seen in Figure \ref%
{peakon_figure} may also be understood as a variant of classical formations
of weak solutions in fluid dynamics by showing that initial conditions exist
for which the solution of the CH equation (\ref{CH-eqn1}) can develop a
vertical slope in its velocity $u(t,x)$, in finite time. The mechanism turns
out to be associated with \emph{inflection points of negative slope}, such
as occur on the leading edge of a rightward-propagating, spatially-confined
velocity profile. In particular, 

\begin{lemma}[Steepening Lemma \protect\cite{CH1993}]
$\quad$\newline
Suppose the initial profile of velocity $u(x,0)$ has an inflection point at $%
x=\overline{x}$ to the right of its maximum, and otherwise it decays to zero
in each direction and that $\|u(\cdot,0)\|_{1,2}<\infty$. Moreover we assume
that $u_x(\bar x,0)<-\sqrt{2M},$ where $M$ is the constant defined in %
\eqref{M}. Then, the negative slope at the inflection point will become
vertical in finite time. 
\index{steepening lemma! CH equation in 1D}
\end{lemma}

\begin{proof}
Consider the evolution of the slope at the inflection point $t\mapsto 
\overline{x}(t)$ that starts at time $0$ from an inflection point $x=%
\overline{x}$ of $u(x,0)$ to the right of its maximum so that 
\begin{equation*}
s_{0}:=u_{x}(\overline{x}(0),0)<\infty .
\end{equation*}%
Define $s_{t}:=u_{x}(\overline{x}(t),t),\ \ t\geq 0$. From the spatial
derivative of the advective form of the CH equation \eqref{CH-eqn3} one
obtains 
\begin{equation*}
\partial _{x}(u_{t}+uu_{x})=-\partial _{x}^{2}K\ast (u^{2}+\frac{1}{2}%
u_{x}^{2})=u^{2}+\frac{1}{2}u_{x}^{2}-K\ast (u^{2}+\frac{1}{2}u_{x}^{2}),
\end{equation*}%
which leads to 
\begin{equation*}
\partial _{t}u_{x}=-uu_{xx}+u^{2}-\frac{1}{2}u_{x}^{2}-K\ast (u^{2}+\frac{1}{%
2}u_{x}^{2}).
\end{equation*}%
This, in turn, yields an equation for the evolution of $t\mapsto s_{t}$.
Namely, by using $u_{xx}(\overline{x}(t),t)=0$ and \eqref{M} one finds 
\begin{eqnarray}
\frac{\mathrm{d}s}{\mathrm{d}t} &=&-\,\frac{1}{2}s^{2}+u^{2}(\overline{x}%
(t),t)-\frac{1}{2}\int_{-\infty }^{\infty }\,\mathrm{e}^{-|\overline{x}%
(t)-y|}\left( u^{2}+\frac{1}{2}u_{y}^{2}\right) \mathrm{d}y  \notag
\label{slope-eqn2} \\
&\leq &-\,\frac{1}{2}s^{2}+M\,.
\end{eqnarray}%
Let $\tilde{s}$ be the solution of the equation 
\begin{equation}
\frac{\mathrm{d}\tilde{s}}{\mathrm{d}t}=-\,\frac{1}{2}\tilde{s}^{2}+M,\ \ 
\tilde{s}_{0}=s_{0}\,.  \label{coth}
\end{equation}%
Observe that 
\begin{equation*}
\frac{\mathrm{d}}{\mathrm{d}t}((s_{t}-\tilde{s}_{t})\mathrm{e}^{{\frac{1}{2}}%
\int_{0}^{t}(s_{p}+\tilde{s}_{p})dp})\leq 0,\ \ s_{0}-\tilde{s}_{0}=0,
\end{equation*}%
therefore, $s_{t}\leq \tilde{s}_{t}$ for all $t>0$ (as long as both are well
defined). However, equation \eqref{coth} admits the explicit solution 
\begin{equation*}
\bar{s}=\sqrt{2M}\coth \left( \sigma +\frac{t}{2}\sqrt{2M}\right)
,\,\,\,\sigma =\coth ^{-1}\left( {\frac{s_{0}}{\sqrt{2M}}}\right) <0\,.
\end{equation*}%
Since $\lim_{t\mapsto -2\sigma /\sqrt{2M}}\bar{s}_{t}=-\infty $ it follows
that there exists a time $\tau \leq -2\sigma /\sqrt{2M}$ by which the slope $%
s_{t}=u_{x}(\overline{x}(t),t)$ becomes negative and vertical, i.e. $%
\lim_{t\mapsto \tau }\bar{s}_{t}=-\infty $. 
\end{proof}

\begin{remark}
Suppose the initial condition is anti-symmetric, so the inflection point at $%
u=0$ is \emph{fixed} and $\mathrm{d}\overline{x}/\mathrm{d}t=0$, due to the
symmetry $(u,x)\to(-u,-x)$ admitted by equation (\ref{CH-eqn1}). In this
case, the total momentum vanishes, i.e. $M=0$, and no matter how small $%
|s(0)|$ (with $s(0)<0$), the verticality $s\to-\infty$ develops at $%
\overline{x}$ in finite time.\smallskip
\end{remark}


\begin{remark}
The Steepening Lemma of \cite{CH1993} proves that in one dimension any
initial velocity distribution whose spatial profile has an inflection point
with negative slope (for example, any antisymmetric smooth initial
distribution of velocity on the real line) will develop a vertical slope in
finite time. Note that the peakon solution (\ref{peakontrain-soln}) has no
inflection points, so it is not subject to the steepening lemma.

The Steepening Lemma underlies the mechanism for forming these singular
solutions, which are continuous but have discontinuous spatial derivatives.
Indeed, the numerical simulations in Figure \ref{peakon_figure} show that
the presence of an inflection point of negative slope in any confined
initial velocity distribution triggers the steepening lemma as the \emph{%
mechanism} for the formation of the peakons. Namely, according to Figure \ref%
{peakon_figure}, the initial (positive) velocity profile ``leans'' to the
right and steepens, then produces a peakon that is taller than the initial
profile, so it propagates away to the right, since the peakon moves at a
speed equal to its height. This process leaves a profile behind with an
inflection point of negative slope; so it repeats, thereby producing a wave
train of peakons with the tallest and fastest ones moving rightward in order
of height. In fact, Figure \ref{peakon_figure} shows that this recurrent
process produces only peakon solutions, as in \eqref{peakontrain-soln}. This
is a result of the isospectral property of CH as a completely integrable
Hamiltonian system, \cite{CH1993}. Namely, the eigenvalues of the initial
profile $u(x,0)$ for the associated CH isospectral problem are equal to the
asymptotic speeds of the peakons in the wave train \eqref{peakontrain-soln}.

The peakon solutions lie in $H ^1$ and have finite energy. We conclude that
solutions with initial conditions in $H^s$ with $s > 3/2$ go to infinity in
the $H^s$ norm in finite time, but they remain in $H^1$ and presumably
continue to exist in a weak sense for all time in $H^1$.
\end{remark}

\color{black}


\section{Advective form of the Stochastic Camassa-Holm (SCH) equation}

Following \cite{Holm2015}, we derive the SCH equation by introducing the
stochastic Hamiltonian function, 
\begin{equation}
\widetilde{h}(m)=\frac{1}{2}\int_{\mathbb{R}}m(x,t)K\ast
m(x,t)\,dx\,dt+\int_{\mathbb{R}}m(x,t)\sum_{i=1}^{N}\xi ^{i}(x)\circ
dW_{t}^{i}\,dx\,.  \label{stoch-Ham}
\end{equation}%
The second term generates spatially correlated random displacements, by
pairing the momentum density with the Stratonovich noise in \eqref{stoch-Ham}
via a set of time-independent prescribed functions $\xi ^{i}(x)$, $%
i=1,2,\dots ,N$, representing the spatial correlations. Thus, the resulting
SCH equation is given by 
\begin{equation}
0=\mathsf{d}m+(\partial _{x}m+m\partial _{x})\frac{\delta \widetilde{h}(m)}{%
\delta m}=\mathsf{d}m+(\partial _{x}m+m\partial _{x}){v}\,,  \label{SCH-eqn}
\end{equation}%
where $m:=u-u_{xx}$ and the stochastic vector field $v$, defined by 
\begin{equation}
v(x,t):=u(x,t)\,dt+\sum_{i=1}^{N}\xi ^{i}(x)\circ dW_{t}^{i}\,,
\label{stochVF}
\end{equation}
represents random spatially correlated shifts in the velocity, \cite%
{Holm2015,HoMaRa1998,HoTyr2016,HoTyr2017}. Thus, the noise introduced in %
\eqref{SCH-eqn} and \eqref{stochVF} represents an additional stochastic
perturbation in the momentum transport velocity.


\subsection{Peakon solutions and isospectrality for the SCH equation}

\begin{theorem}
The SCH equation \eqref{SCH-eqn} with the stochastic vector field $v$ in %
\eqref{stochVF} admits the singular momentum solution for CH in %
\eqref{delta-def} for peakon wave trains.
\end{theorem}

\begin{proof}
Substituting the singular momentum relation \eqref{delta-def} into the stochastic Hamiltonian $\widetilde{h}(m)$ in \eqref{stoch-Ham} and performing the integrals yields
the Hamiltonian for the stochastic peakon trajectories as
\begin{align*}
\widetilde{h}(q,p)
&:=
\frac14\sum_{a,b=1}^M p_a(t)p_b(t) e^{-|q_a(t) - q_b(t)|}
+
\sum_{a,b=1}^M p_a(t)\sum_{i=1}^N \xi ^{i}(q_a(t))\circ dW_{t}^{i}
\,.
\end{align*}%
The canonical Hamiltonian equations for the stochastic peakon trajectories and momenta are thus given by 
\begin{align*}
\mathsf{d}q_a &= \frac{\partial\widetilde{h}}{\partial p_a}
=
\frac12\sum_{b=1}^M p_b(t) e^{-|q_a(t) - q_b(t)|}\,dt
+
\sum_{i=1}^N \xi ^{i}(q_a(t))\circ dW_{t}^{i}
\\&=
u(q_a(t))\,dt
+
\sum_{i=1}^N \xi ^{i}(q_a(t))\circ dW_{t}^{i}
=
v(q_a(t))
\,,
\end{align*}%
and
\begin{align*}
\mathsf{d}p_a = -\,\frac{\partial\widetilde{h}}{\partial q_a}
=
- \,p_a(t) \frac{\partial u}{\partial q_a}\,dt
- p_a(t)
\sum_{i=1}^N \frac{\partial \xi ^{i}}{\partial q_a}\circ dW_{t}^{i}
=
- \,p_a(t) \frac{\partial v(q_a(t))}{\partial q_a}
\,. 
\end{align*}%
Substituting these stochastic canonical Hamiltonian equations for $q_a(t)$ and $p_a(t)$ into the singular momentum solution for CH in \eqref{delta-def} recovers the SCH equation \eqref{SCH-eqn} and the stochastic vector field $v$ in \eqref{stochVF}.
\end{proof}

Thus, the SCH equation \eqref{SCH-eqn} admits peakon wave train solutions
whose peaks in velocity follow the stochastic trajectories given by the
stochastic vector field $v$ in \eqref{stochVF} and satisfy stochastic
canonical Hamiltonian equations. The corresponding canonical Hamiltonian
equations in the absence of noise describe the trajectories and momenta of
CH wave trains. For numerical studies of the interactions of stochastic
peakon solutions, see \cite{HoTyr2016,HoTyr2017}.

Remarkably, a certain amount of the isospectral structure for the
deterministic CH equation is preserved by the addition of the stochastic
transport perturbation we have introduced in \eqref{SCH-eqn} and %
\eqref{stochVF}.

\begin{theorem}[Isospectral problem for SCH]
\label{SCH-int-thm} The SCH equation in \eqref{SCH-eqn} follows from the
compatibility condition for the deterministic CH isospectral eigenvalue
problem \eqref{CHevprob}, and a stochastic evolution equation for the real
eigenfunction $\psi$, 
\begin{align}
\psi_{xx} &= \left(\frac14 - \frac{m}{2\lambda} \right)\psi \,,
\label{SCHevprob} \\
\mathsf{d}\psi &= -(\lambda + v) \psi_x + \frac12 v_x\psi \,,
\label{SCHevoleqn} \\
\hbox{with}\quad v:&= u\,dt + \sum_{i=1}^{N}\xi ^{i}(x) \circ dW^i_t \,,
\label{stochVFthm}
\end{align}
and real isospectral parameter, $\lambda$, provided $\mathsf{d}\lambda=0$
and $\xi^{i}(x) = C^{i}+A^{i}e^x+B^{i}e^{-x}$, for constants $A^{i},B^{i}$
and $C^{i}$.
\end{theorem}

\begin{proof}
By direct calculation, equating cross derivatives $\mathsf{d}\psi_{xx} = \partial_x^2\mathsf{d}\psi$ using equations \eqref{SCHevprob} and \eqref{SCHevoleqn} implies, when $\mathsf{d}\lambda=0$, that
\[
\mathsf{d}m+(\partial _{x}m+m\partial _{x}){v}  
+ \lambda \big(m_x - (v_x-v_{xxx})\big) = 0\,.
\]
Consequently, the compatibility condition for equations \eqref{SCHevprob} and \eqref{SCHevoleqn} implies the SCH equation in \eqref{SCH-eqn}, provided $\mathsf{d}\lambda=0$ and $\xi_x^{i}(x)-\xi_{xxx}^{i}(x) = 0$. The latter means that $\xi^{i}(x)$ is either constant, or exponential. 
\end{proof}

\begin{remark}
Theorem \ref{SCH-int-thm} means that the SCH equation \eqref{SCH-eqn} with
stochastic vector field $v$ \eqref{stochVF} with $\xi_x^{i}(x)-%
\xi_{xxx}^{i}(x) = 0$ has the same countably infinite set of conservation
laws as for the deterministic CH equation \eqref{CH-eqn1}. However, the SCH
Hamiltonian $\widetilde{h}(m)$ in \eqref{stoch-Ham} is not conserved by the
SCH equation \eqref{SCH-eqn}, because it depends explicitly on time.
Consequently, for those choices of $\xi^{i}(x)$, the SCH equation is
equivalent to consistency of the linear equations \eqref{SCHevprob} and %
\eqref{SCHevoleqn}. Therefore, SCH is solvable by the isospectral method for
each realisation of the stochastic process in \eqref{stochVFthm}. However,
SCH is almost certainly not completely integrable as a Hamiltonian system,
even for constant $\xi^{i}$. See \cite{Ar2015} for an example of an
integrable stochastic deformation of the CH equation.
\end{remark}

The issue now and for the remainder of the paper is to find out whether the
wave breaking property which is the mechanism for the creation of peakon
wave trains in the deterministic case also survives the introduction of
stochasticity.


\subsection{Wave breaking estimates for SCH}

In the following we will assume the conditions under which the stochastic
integrals appearing in equation \eqref{SCH-eqn} for $u$ as well as the
equation for $u_{x}$ are well defined and summable. In particular, we assume
that the vector fields $\xi _{i}$ are smooth and bounded and that 
\begin{equation*}
\sum_{i>0}((\Vert \xi ^{i}\Vert _{\infty })^{2}+(\Vert \xi _{x}^{i}\Vert
_{\infty })^{2}+(\Vert \xi ^{i}\Vert _{2,1})^{2})<\infty \,.
\end{equation*}%
Let $A^{i},\partial _{x}A^{i},$ $i\in \mathbb{Z}_{+}$ be the following set
of operators 
\begin{eqnarray}
A^{i}(u) &=&u_{x}\xi ^{i}-K\ast \left( u_{x}\xi _{xx}^{i}(x)+2u\xi
_{x}^{i}(x)\right) ,\   \label{ai} \\
\partial _{x}A^{i}(u) &=&u_{xx}\xi ^{i}+u_{x}\xi _{x}^{i}-\partial _{x}K\ast
\left( u_{x}\xi _{xx}^{i}(x)+2u\xi _{x}^{i}(x)\right) ,  \label{dxai}
\end{eqnarray}%
( $\partial _{x}A^{i}$ is obtained by formally differentiating $A^{i}$). In
the following we will assume that there is a local solution of equation %
\eqref{SCH-eqn} such that the operators $A^{i},\partial _{x}A^{i},$ $i\in 
\mathbb{Z}_{+}$ are well defined.

Let us deduce first the equation for the velocity slope, $u_x$. We have the
following Lemma:

\begin{lemma}[Evolution of the velocity slope]
Under the above conditions, we have 
\begin{equation}
\mathsf{d}u_{x}=-\frac{1}{2}\left( u_{x}^{2}+2uu_{xx}-u^{2}\right)
\,dt-K\ast \left( u^{2}+\frac{1}{2}u_{x}^{2}\right)
dt-\sum_{i}A_{x}^{i}(u)\circ dW_{t}^{i}.  \label{ux}
\end{equation}
\end{lemma}

\begin{proof}

Expanding out the SCH equation in terms of $u$ and $v$ gives
\begin{equation*}
\begin{split}
0& =\mathsf{d}m+(\partial _{x}m+m\partial _{x})v \\
& =(1-\partial _{x}^{2})\mathsf{d}u+2uv_{x}+u_{x}v-2v_{x}u_{xx}-vu_{xxx} \\
& =(1-\partial _{x}^{2})(\mathsf{d}u+vu_{x})+u_{x}v_{xx}+2uv_{x} \\
& =(1-\partial _{x}^{2})(\mathsf{d}u+vu_{x})+\partial _{x}\left( u^{2}+\frac{%
1}{2}u_{x}^{2}\right) \,dt 
+\sum \left( u_{x}\xi _{xx}^{i}(x)+2u\xi _{x}^{i}(x)\right) \circ
dW_{t}^{i}\,.\quad
\end{split}%
\end{equation*}%
Therefore, applying the smoothing operator $K\ast :=(1-\partial
_{x}^{2})^{-1}$, given by the convolution with the Green's function $K(x,y)$ in \eqref{Greens-relation} for
the Helmholtz operator $(1-\partial _{x}^{2})$, to both sides of the
previous equation yields
\begin{eqnarray}
\mathsf{d}u+uu_{x}\,dt &=&-u_{x}\left( \sum \xi ^{i}\circ dW_{t}^{i}\right)
-\partial _{x}K\ast \left( u^{2}+\frac{1}{2}u_{x}^{2}\right) \,dt+\sum K\ast \left( u_{x}\xi _{xx}^{i}+2u\xi _{x}^{i}\right) \circ
dW_{t}^{i} \nonumber \\
&=& -\partial _{x}K\ast \left( u^{2}+\frac{1}{2}u_{x}^{2}\right) dt-\sum (u_{x}\xi ^{i}-K\ast \left( u_{x}\xi _{xx}^{i}(x)+2u\xi
_{x}^{i}(x)\right))\circ dW_{t}^{i}\nonumber \\
&=&-\partial _{x}K\ast \left( u^{2}+\frac{1}{2}u_{x}^{2}\right) dt-\sum
A^{i}(u)\circ dW_{t}^{i},\label{equationforu}
\end{eqnarray}%
in which the derivative $\partial_x$ is understood to act on everything standing to its right. 
Consequently, we have 
\begin{equation}
\mathsf{d}u_{x}=-\left( u_{x}^{2}+uu_{xx}\right) \,dt-\partial _{xx}K\ast
\left( u^{2}+\frac{1}{2}u_{x}^{2}\right) dt-\sum A_{x}^{i}(u)\circ
dW_{t}^{i}. \label{uxi}
\end{equation}%
Then, since 
\[
\partial _{xx}K\ast \left( u^{2}+\frac{1}{2}u_{x}^{2}\right)  =-\left( u^{2}+\frac{1}{2}u_{x}^{2}\right) +K\ast \left( u^{2}+\frac{1}{2}u_{x}^{2}\right),    \label{partialxk} 
\]
we deduce \eqref{ux}.
\end{proof}

\begin{remark}
Observe that 
\begin{eqnarray}
\partial _{x}K\ast \left( u_{x}\xi _{xx}^{i}(x)+2u\xi _{x}^{i}(x)\right)
&=&\partial _{xx}K\ast (u\xi _{xx}^{i}(x))-\partial _{x}K\ast \left( u\xi
_{xxx}^{i}(x)+2u\xi _{x}^{i}(x)\right)  \notag \\
&=& -u\xi _{xx}^{i}(x)+K\ast (u\xi _{xx}^{i}(x))-\partial _{x}K\ast \left(
u\xi _{xxx}^{i}(x)+2u\xi _{x}^{i}(x)\right).  \label{partialxk2}
\end{eqnarray}%
Hence, the last term in the expression of \eqref{dxai} can be controlled by
the supremum norm of $u$.
\end{remark}

Just as in the deterministic case, we define next the process $t\mapsto \nu
_{t}$ as the inflection point of $u$ to the right of its maximum so that 
\begin{equation*}
u_{xx}\left( \nu _{t},t\right) =0,\ \ \mathrm{and}\ \ s_{t}=u_{x}\left( \nu
_{t},t\right) <0.
\end{equation*}%
In what follows, we will assume, without proof, that the process $t\mapsto
\nu _{t}$ is a semi-martingale. The argument to show that validity of this
property is based on the implicit function theorem. Indeed one can show that 
$\nu $ satisfies the equation 
\begin{equation*}
d\nu _{t}=-\,{\frac{1}{u_{xxx}(\nu _{t},t)}}(\mathsf{d}u_{xx})(\nu _{t},t)\,,
\end{equation*}%
provided the equation is well defined. That is, assume we have an inflection
point, not an inflection interval; which means we have assumed $u_{xxx}(\nu
_{t},t)\neq 0$. Using the semimartingale property of $\nu $ and the It\^{o}%
-Ventzell formula (see, e.g., \cite{op}), we deduce that 
\begin{equation*}
\mathsf{d}\left( u_{x}(\nu _{t},t)\right) =\left( \mathsf{d}u_{x}\right)
\left( \nu _{t},t\right) +u_{xx}\left( \nu _{t},t\right) \circ \mathsf{d}\nu
_{t}=\left( \mathsf{d}u_{x}\right) \left( \nu _{t},t\right) .
\end{equation*}%
Hence, by \eqref{ux} and \eqref{partialxk2}, we find that 
\begin{eqnarray}
ds_{t} &=&-\left( \frac{1}{2}s_{t}^{2}-u^{2}\left( \nu _{t},t\right) \right)
dt-K\ast \left( u^{2}+\frac{1}{2}u_{x}^{2}\right) \left( \nu _{t}\right) 
\notag \\
&&-\sum_{i}\left( s_{t}\xi _{\nu _{t}}^{i}+B^{i}(u)|_{\nu _{t}}\right) \circ
dW_{t}^{i},  \label{st}
\end{eqnarray}%
where the operators $B^{i}$ are given by 
\begin{equation*}
B^{i}(u)=-u\xi _{xx}^{i}(x)+K\ast (u\xi _{xx}^{i}(x))-\partial _{x}K\ast
\left( u\xi _{xxx}^{i}(x)+2u\xi _{x}^{i}(x)\right) ,\ \ i\in \mathbb{Z}_{+}
\end{equation*}%
We will henceforth consider the particular case when the vector fields $\xi
^{i}$ are spatially homogeneous, so that $B^{i}(u)=0$. (This is also the
isospectral case, which we discussed in the previous section.) In this case,
just as in the deterministic case, we have 
\begin{equation}
\Vert u(\cdot ,t)\Vert _{1,2}=\Vert u(\cdot ,0)\Vert _{1,2},\ \ \ \mathrm{%
for\ \ all}\ \ x\in \mathbb{R}.  \label{constantnormsto}
\end{equation}%
and, again, \eqref{constantnormsto} implies that 
\begin{equation}
M:=\sup_{t\in \lbrack 0,\infty )}\Vert u(\cdot ,t)\Vert _{\infty }<\infty .
\label{MM}
\end{equation}%
This bound arises because the stochastic term vanishes when computing $%
\mathrm{d}\Vert u(\cdot ,t)\Vert _{1,2}^{2}$ . More precisely, the
stochastic term is given by the expression 
\begin{equation*}
\sum (2uu_{x}+u_{x}u_{xx})\,\xi ^{i}\circ W_{t}^{i},
\end{equation*}%
whose spatial integral over the real line vanishes for constant $\xi ^{i}$,
for the class of solutions $u(\cdot ,t)$ which vanish at infinity and whose
gradient also vanishes at infinity. Note that the constant $M$ in \eqref{MM}
is independent of the realization of the Brownian motions $W^{i}$, $i\in {%
\mathbb{Z}_{+}}$. By a standard Sobolev embedding theorem, %
\eqref{constantnorm} also implies the useful relation that 
\begin{equation}
M:=\sup_{t\in \lbrack 0,\infty )}\Vert u(\cdot ,t)\Vert _{\infty }<\infty .
\end{equation}

\begin{proposition}
\label{p9}As in the deterministic case, suppose the initial profile of
velocity $u(x,0)$ has an inflection point at $x=\overline{x}$ to the right
of its maximum, and it decays to zero in each direction; so that $%
\|u(\cdot,0)\|_{1,2}<\infty$. Consider the expectation of the slope at the
inflection point, $\bar{s}_{t}=E\left[ s_{t}\right] $. If $u_x(\bar x,0)$ is
sufficiently small, then there exists $\tau<\infty$ such that $%
\lim_{t\mapsto \tau} \bar{s}_{t}=-\infty$.
\end{proposition}

\begin{proof}

By changing from Stratonovitch to It\^{o} integration, we obtain from %
\eqref{st} that 
\begin{equation}  \label{stochasticcase}
ds_{t}=-\left( \frac{1}{2}s_{t}^{2}-u^{2}\left( \nu _{t},t\right) \right)
dt-K\ast \left( u^{2}+\frac{1}{2}u_{x}^{2}\right) \left( \nu _{t}\right)
-\sum_{i}s_{t}\xi ^{i}dW_{t}^{i}+\frac{1}{2}\sum_{i}s_{t}\left( \xi
^{i}\right) ^{2}
\end{equation}%
and, by taking expectation, we deduce that 
\begin{equation*}
dE\left[ s_{t}\right] \leq -\frac{1}{2}\left( E\left[ s_{t}^{2}\right] -E%
\left[ s_{t}\right] ^{2}\right) -\frac{1}{2}\left( E\left[ s_{t}\right]
^{2}\right) +\frac{\left\vert \left\vert \xi \right\vert \right\vert }{2}E%
\left[ s_{t}\right] +M
\end{equation*}%
Consequently, 
\begin{equation*}
d\bar{s}_{t}\leq -\frac{1}{2}\left( \bar{s}_{t}\right) ^{2}+\frac{\left\vert
\left\vert \xi \right\vert \right\vert }{2}\bar{s}_{t}+M\leq -\frac{%
1-\varepsilon }{2}\left( \bar{s}_{t}\right) ^{2}+\left( M+\frac{\left\vert
\left\vert \xi \right\vert \right\vert ^{2}}{2\varepsilon }\right)
\end{equation*}%
from which we deduce that the magnitude $|\bar{s}_{t}|$ blows up in finite
time, just as in the deterministic case. 
\end{proof}

Note that Proposition \ref{p9} does not guarantee pathwise blow up of the
process for the magnitude of the negative slope at the inflection point $%
|s_t|$, only the blow up of its mean, $|\bar{s}_{t}|$. The following theorem
shows that, indeed the pathwise negative slope $s_t$ blows up in finite time
with positive probability, albeit not with probability 1.

\begin{theorem}[Wave breaking for the Stochastic Camassa-Holm equation]
Under the same assumptions as those introduced in Proposition \ref{p9}, with
positive probability, the negative slope at the inflection point $%
s_t=u_{x}(\nu_t,t)$ will become vertical in finite time.
\end{theorem}

\begin{proof}

We define a new Brownian motion $W$, as follows 
\begin{equation*}
W_{t}=\frac{-\sum_{i}\xi ^{i}W_{t}^{i}}{\left\vert \left\vert \xi
\right\vert \right\vert },~~~~t>0.
\end{equation*}%
Then the equation for $s_{t}$ becomes%
\begin{equation*}
ds_{t}=-\left( \frac{1}{2}s_{t}^{2}-\frac{\left\vert \left\vert \xi
\right\vert \right\vert ^{2}}{2}s_{t}-u^{2}\left( \nu _{t},t\right) \right)
dt-K\ast \left( u^{2}+\frac{1}{2}u_{x}^{2}\right) \left( \nu _{t}\right)
dt+\left\vert \left\vert \xi \right\vert \right\vert s_{t}dW_{t}\,.
\end{equation*}%
We introduce a Brownian motion $B$ such that the stochastic integral $%
\int_{0}^{t}s_{p}{d}W_{p}$ can be represented as 
\begin{equation*}
\int_{0}^{t}s_{p}{d}W_{p}=B_{\int_{0}^{t}s_{p}^{2}dp}\,.
\end{equation*}%
Then, as above, 
\begin{eqnarray*}
s_{t} &\leq &s_{0}+\frac{1}{2}\int_{0}^{t}\left( \left( M+\frac{\left\vert
\left\vert \xi \right\vert \right\vert ^{2}}{2\varepsilon }\right) +\frac{%
\varepsilon s_{p}^{2}}{2}\right) dp+\left( -\frac{1-2\varepsilon }{2}%
\int_{0}^{t}s_{p}^{2}dp+\left\vert \left\vert \xi \right\vert \right\vert
B_{\int_{0}^{t}s_{p}^{2}dp}\right)  \\
&\leq &s_{0}+\frac{1}{2}\int_{0}^{t}\left( M+\frac{\left\vert \left\vert \xi
\right\vert \right\vert ^{2}}{2\varepsilon }\right)
dp+X_{\int_{0}^{t}s_{p}^{2}dp}
\end{eqnarray*}%
where\ $X\ $is a Brownian motion with negative drift\footnote{%
If $X$ be a Brownian motion with negative drift, $X(t)=\sigma B(t)+\mu t$, $%
\mu <0,$ $\lim_{t\mapsto \infty }X\left( t\right) =-\infty .$ Let $%
M=\max_{s\geq 0}X\left( t\right) .$ Then $P\left( M\geq a\right) =\exp
\left( {-a\left( \frac{2\left\vert \mu \right\vert }{\sigma ^{2}}\right) }%
\right) $ and we conclude that M has an exponential distribution with mean $%
\frac{2\left\vert \mu \right\vert }{\sigma ^{2}}$. Put it differently, no
matter where we start the Brownian motion with drift there is a positive
probability that it will reach any level, before it drifts off to $-\infty $%
. Vice versa, it never hits level $a$ with positive probability, see e.g. 
\cite{KS}.} 
\begin{equation*}
X\left( t\right) =-\frac{1-3\varepsilon }{2}t+\left\vert \left\vert \xi
\right\vert \right\vert B_{t}.
\end{equation*}%
\textbf{With positive probability (though not 1!),} the process 
\begin{equation*}
t\mapsto X_{\int_{0}^{t}s_{p}^{2}dp}
\end{equation*}%
remains smaller than, say, $s_{0}/2<0$ for all $t>0$. If the magnitude of
the negative slope at the inflection point $|s_{0}|$ is sufficiently large,
then \textbf{with positive probability } the term 
\begin{equation*}
\frac{1}{2}\int_{0}^{t}\left( \left( M+\frac{\left\vert \left\vert \xi
\right\vert \right\vert ^{2}}{2\varepsilon }\right) -\frac{\varepsilon
s_{p}^{2}}{2}\right) ds
\end{equation*}%
will always stay negative. It follows that, with positive probability, we
have 
\begin{equation*}
\lim_{t\mapsto \infty }\int_{0}^{t}s_{p}^{2}dp=\infty ,\quad \hbox{and}\quad
\lim_{t\mapsto \infty }s_{t}\leq \lim_{t\mapsto
0}(s_{0}+X_{\int_{0}^{t}s_{p}^{2}dp})=-\infty \,.
\end{equation*}%
Moreover, for sufficiently large $t$, 
\begin{equation*}
X_{\int_{0}^{t}s_{p}^{2}dp}=-\frac{1-2\varepsilon }{2}%
\int_{0}^{t}s_{p}^{2}dp+\left\vert \left\vert \xi \right\vert \right\vert
B_{\int_{0}^{t}s_{p}^{2}dp}\leq -\frac{1-3\varepsilon }{2}%
\int_{0}^{t}s_{p}^{2}dp\,.
\end{equation*}%
Hence, 
\begin{equation*}
s_{t}\leq s_{0}-\frac{1-3\varepsilon }{2}\int_{0}^{t}s_{p}^{2}dp.
\end{equation*}%
which, in turn, implies, as in the deterministic case, that the negative
slope $s_{t}$ at the inflection point must become vertical in finite time. 
\end{proof}

\begin{remark}
In a similar manner we can show that $s_{t}\geq \tilde{s}_{t},$ $t\geq 0,$
where%
\begin{equation*}
\tilde{s}_{t}=-\left( \frac{1}{2}\tilde{s}_{t}^{2}-\frac{\left\vert
\left\vert \xi \right\vert \right\vert ^{2}}{2}\tilde{s}_{t}+M\right)
dt+\left\vert \left\vert \xi \right\vert \right\vert \tilde{s}%
_{t}dW_{t}\,,~~~\tilde{s}_{0}=s_{0}.
\end{equation*}%
To show this one proceeds, as in the proof of the steepening lemma, by first 
justifying the inequality  
\begin{equation*}
\frac{\mathrm{d}}{\mathrm{d}t}((s_{t}-\tilde{s}_{t})\mathrm{e}^{{\frac{1}{2}}%
\int_{0}^{t}(s_{p}+\tilde{s}_{p})dp+\frac{\left\vert \left\vert \xi
\right\vert \right\vert ^{2}t}{2}+\left\vert \left\vert \xi \right\vert
\right\vert W_{t}})\geq 0,\ \ s_{0}-\tilde{s}_{0}=0,
\end{equation*}%
so that, $s_{t}\geq \tilde{s}_{t}$ for all $t>0$ (as long as both are well
defined). In turn, $\tilde{s}_{t}$, and therefore $s_{t}$, may achieve
positive values with positive probability{, which could, in principle, lead to a violation of the conditions under which a peakon may emerge in finite time due to the presence of an inflection point with slope $s_{t}$.} Future work is planned by the authors to further investigate the emergence of peakons as well as the local well-posedness of the
stochastic CH\ equation.
\end{remark}

\end{document}